\documentclass[twoside,12pt,a4paper]{article}
\usepackage{amsmath,float,amssymb,mathtools,amsthm,mathrsfs,multirow,booktabs,setspace,morefloats,epsfig,caption,subfigure,graphicx,subfig,subfigure,placeins}
\usepackage{natbib}
\makeatletter
\newcommand{\thickhline}{%
	\noalign {\ifnum 0=`}\fi \hrule height 1pt
	\futurelet \reserved@a \@xhline
}
\usepackage{fancyhdr}
\usepackage{lipsum}

\newcommand\shorttitle{Some Goodness of Fit Tests based on Centre Outward Spacings}
\newcommand\authors{Rahul Singh}

\title{Some Goodness of Fit Tests based on Centre Outward Spacings}
\author{Rahul Singh\\ {\small Email:~sirahul@iitk.ac.in}}
\date{
	\small{Department of Mathematics and Statistics, Indian Institute of Technology Kanpur, India}\\%
}

\newtheorem{lemma}{Lemma}
\newtheorem{theorem}{Theorem}
\newtheorem{definition}{Definition}
\newtheorem{remark}{Remark}
\allowdisplaybreaks
\begin{document}

	\maketitle
		
		\begin{abstract}
			Data depth provides a centre-outward ordering for multivariate data. Recently, some univariate GoF tests based on data depth have been studied by \cite{li2018}. This paper discusses some univariate goodness of fit tests based on centre-outward spacings. These tests have similar asymptotic properties (distribution and efficiency) as those based on usual spacings. A simulation study reveals that for light-tailed symmetric alternatives, the proposed tests perform better than those based on usual spacings.
		\end{abstract}

	
	\section{Introduction}
	For absolutely continuous distribution functions (dfs), a popular method of univariate goodness of fit (GoF) tests is based on sample spacings. Let $X_1,\ldots,X_{n-1}$ be independent and identically distributed (i.i.d.) random variables from an absolutely continuous df $F$. Let $X_{(1)}\! \leq \cdots \leq \!X_{(n-1)}$ denote the corresponding order statistics. Define $X_{(0)}=-\infty$ and $X_{(n)}=\infty$. The $m$-step spacings are defined as $D_k^{(m)}:= F(X_{(k+m-1)})- F(X_{(k-1)})$ for $k=1,2,\ldots,n-m+1$. For $m=1$, these are known as simple spacings, usually denoted by $D_k$'s. A typical GoF test statistic based on spacings has the form $W(h):=\frac{1}{n}\sum _{i=1}^n h(nD_i)$, where $h$ is some convex function. Some popular choices of function $h$ are as follows:
	\begin{table*}[h!]
		\begin{center}
			\begin{tabular}{l l} 
				\hline
				$h(x)$ & Statistic \\
				\hline 
				$x^2$ & Greenwood Statistic (\citealt{green}) \\
				
				$-\log(x)$ & Log Spacing Statistic (\citealt{moran1951})\\
				
				$|x-1|$ & Rao's Spacing Statistic (\citealt{rao:spacing:test})\\
				
				$x\log( x)$ & Relative Entropy Spacing Statistic (\citealt{misra2001new})\\
				\hline
			\end{tabular}
		\end{center}
	\end{table*}
	\FloatBarrier
	\noindent
	Such a statistic is an estimator of a $\phi$-divergence and a natural candidate for a GoF test statistic. 
	Initially, \cite{sethu1970} and \cite{rao1975weak} discovered that a class of such statistics are asymptotically normal under simple null and a smooth sequence of alternative converging to null at the rate of $n^{1/4}$. They found that the Greenwood test is asymptotically the most efficient in terms of the Pitman asymptotic relative efficiency (ARE) for this sequence of alternatives. Using the same approach,  \cite{del1979asymptotic} found that the Greenwood type test based on disjoint $m$-step spacings is asymptotically more efficient than the usual Greenwood test. \cite{rao-kuo1984} observed that, for the fixed step $m$, tests based on overlapping spacing are asymptotically more efficient than the corresponding tests based on disjoint spacings and the Greenwood type test is asymptotically most efficient among tests based on symmetric functions of overlapping $m$-step spacings. 
	
	In the multivariate statistics literature, the data depth of a point is a measure of centrality of the point with respect to the data cloud or the underlying df. There are various notions of data depth (see e.g., \citealt{zuo2000general}). Two popular notions are half-space depth (\citealt{tukey1975mathematics}) and simplicial depth (\citealt{liu1990notion}). In fact, a data depth induces centre-outward (CO) ordering. This ordering in the univariate case can be utilised for GoF tests. Recently, \cite{li2018} studied Kolmogorov-Smirnov, Anderson-Darling, Cramer von-Mises tests based on CO ordering. They found that the GoF tests based on CO ordering perform better than their usual counterparts for alternatives with scale differences. 
	
	In this paper, we define sample spacings based on CO ordering and study GoF tests based on such spacings. Such tests have not been studied in the literature and are of potential theoretical and practical interest. We also perform a small simulation study. The aim of this simulation study is to compare performances of the proposed tests and GoF tests based on usual spacings.

	\section{Centre-Outward Spacings}
	For the univariate case, let $S\equiv 1-F$ denote the survival function.  Then, the half-space depth and the simplicial depth of a point $x\in\mathbb{R}$ with respect to the df $F$ are given by
	$\min(F(x),S(x))$  and  $2F(x)S(x)$, respectively. 
	For the univariate case, half-space depth and simplicial depth achieve maximum at the median  of the df and monotonically decrease to zero on either side of median. So, we can use either of them to construct CO ordering of observations. \cite{li2018} discussed both the univariate half-space and simplicial depths, and found that they provide the same CO ordering. 
	
	Denote the depth (half-space, or simplicial) with respect to the df $F$ by $D_F$. Define $R_Y= P_F[D_F(X)\geq D_F(Y)| Y]$ for $X\sim F$. Then, $R_X=|2F(X)-1|$, and $R_{X_i}\stackrel{i.i.d.}{\sim} U(0,1)$ for $i=1,2,\ldots,n-1$ (see \citealt{li2018}). 
	Note that $R_Y$ is a decreasing function of $D_F(Y)$. Let $R^{(1)},R^{(2)},\ldots,R^{(n-1)}$ be the order statistics corresponding to $R_{X_1},R_{X_2},\ldots,R_{X_{n-1}}$, $R^{(0)}=0$ and $R^{(n)}=0$. Now, we can define sample spacings based on $R_{X_i}$s. We call these spacings the ``CO spacings''. 
	\begin{definition} 
		Under the above described set-up,  we define CO spacings as  
		\begin{align*}
			DS_i=R^{(i)}-R^{(i-1)} \text{ for }i=1,2,\ldots,n.
		\end{align*}
	\end{definition}
	The following result gives the distribution of CO spacings.
	\begin{lemma}
		For an absolutely continuous df $F$, we have
		\begin{align*}
			(DS_1,DS_2,\ldots,DS_n)\stackrel{d}{=}(T_1,T_2,\ldots,T_n),
		\end{align*} 
		where $(T_1,T_2,\ldots,T_n)$ are simple spacings corresponding to a random sample of size $n-1$ from the $U(0,1)$~df. 
	\end{lemma}
	This  result is a consequence of the fact that $R_{X_i}\stackrel{i.i.d.}{\sim} U(0,1)$ for $i=1,2,\ldots,n-1$. Thus, the CO spacings have the same distribution as the usual spacings. 
	
	\section{Goodness of Fit Tests based on CO Spacings}
	The goal is to test $H_0: F=F_0$ against $H_1:F\neq F_0$, where $F_0$ is a completely specified df. Using the probability integral transform, this is equivalent to testing uniformity, i.e., $H_0: F(x)=x~\forall\, x\in [0,1]$ against $H_1: F(x)\neq x \text{ for some } x\in [0,1]$, where the support of $F$ is $[0,1]$. 
	Under $H_0$, the CO ordering random variable is $R_X=|2X-1|$. For $X\sim F$, the df of $R_X$ is as follows:
	$$F_R(y)=P(R_X\leq y) 
	=\begin{cases}
		0, &\text{ if } y<0,\\
		F\left(\frac{1+y}{2}\right)-F\left(\frac{1-y}{2}\right), &\text{ if } y\in [0,1],\\
		1, &\text{ if } y> 1.
	\end{cases}$$
	Denote the density function of $F$ by $f$. Then the density function of $R_X$ is given by
	$$ f_R(y)= \begin{cases}
		\frac{1}{2}\left( f(\frac{1+y}{2})+f(\frac{1-y}{2})\right), &\text{ if } y\in [0,1]\\
		0, &\text{ otherwise. }  
	\end{cases}$$
	
	Let $F_1$ and $F_2$ be two dfs with corresponding density functions $f_1$ and $f_2$, respectively. Then, the Hellinger distance (HD) between the dfs $F_1$ and $F_2$ is defined as  $HD(F_1,F_2)=\sqrt{1-\int_\mathbb{R}\sqrt{f_1(x)f_2(x)}\, dx}$. 
	\begin{lemma}\label{cos:lemma2}
		Let $F_0$ denote the df of $U(0,1)$, $X\sim F$ and $R_X=|2X-1|$. Denote the df of $R_X$ by $F_R$. Then, $HD(F_0,F_R) \leq HD(F_0,F)$. Moreover, if $X$ is symmetric about $1/2$, then $HD(F_0,F_R) = HD(F_0,F)$.
	\end{lemma} 
	\begin{proof}
		Observe that $\sqrt{2}\sqrt{a+b}\geq \sqrt{a}+\sqrt{b}$ for $a,b\geq 0$, and equality holds iff $a=b$. Here, $HD(F_0,R_X) = 1-\int_0^1\sqrt{f_R(y)} dy$.
		\begin{align*}
			\int_0^1\sqrt{f_R(y)} dy &= \frac{1}{\sqrt{2}}\int_0^1\sqrt{f\left(\frac{1+y}{2}\right)+f\left(\frac{1-y}{2}\right)} dy\\
			& \geq \frac{1}{2}\left[\int_0^1\sqrt{f\left(\frac{1+y}{2}\right)}dy+ \int_0^1\sqrt{f\left(\frac{1-y}{2}\right)}dy\right]\\
			& = \frac{1}{2}\left[2\int_{0.5}^1\sqrt{f(x)}dx+ 2\int_0^{0.5}\sqrt{f(x)}dx\right]
			= \int_0^1 \sqrt{f(x)}dx.
		\end{align*}
		Hence, $HD(F_0,F_R) \leq HD(F_0,F)$ and equality holds if $f\left(\dfrac{1+y}{2}\right)=f\left(\dfrac{1-y}{2}\right) \forall\, y\in [0,1]$, i.e., $f$ is symmetric about $1/2$.
	\end{proof}
	Lemma \ref{cos:lemma2} suggests that, when the underlying distribution is not symmetric, the  HD between the df of CO ordering random variable and $U(0,1)$ df is less than the HD between the underlying df and  $U(0,1)$ df. This explains why CO ordering based GoF tests have low power in detecting location differences, which was also observed by \cite{li2018} in simulation studies.
	
	Inspired by GoF tests based on usual spacings, we propose the following class of GoF test statistics based on CO spacings
	\begin{align*}
		W^*(h)= \frac{1}{n}\sum_{i=1}^nh(nDS_i).
	\end{align*}
	Note that these test statistics based on CO spacings are distribution-free and have the same distribution as the corresponding test statistics based on usual spacings. So, a test based on CO spacings has the same critical values as corresponding usual spacings based test.
	\subsection{Some Asymptotic Results}
	
	We consider test statistics based on CO spacings of type:  
	$ W^*(h)= \frac{1}{n}\sum_{i=1}^nh(nDS_i)$, where $h$ satisfies assumption (3.3) of \cite{del1979asymptotic}. Following \cite{sethu1970}, we consider sequence of local alternatives of the type
	\begin{align}\label{cos:alternative}
		F_n(x)= x+ \dfrac{L_n(x)}{\sqrt[4]{n}} \text{ for } 0\leq x\leq 1,
	\end{align}
	where $L_n(0)=L_n(1)=0,$ $L_n(x)$ is twice differentiable on the unit interval [0,1]. Further, assume that there exist a function $L(x)$ which is twice continuously differentiable with $L(0)=L(1)=0$, such that
	\[\sqrt[4]{n}\sup_{\substack{0\leq x \leq 1}}|L_n(x)-L(x)|= o(1),\\
	\sqrt[4]{n}\sup_{\substack{0\leq x \leq 1}}|L_n'(x)-L'(x)|= o(1)\\
	\text{ and } \sqrt[4]{n}\sup_{\substack{0\leq x \leq 1}}|L_n''(x)-L''(x)|= o(1).\]
	For the above mentioned sequence of local alternatives, the df of $R_X$ is given by
	\[F_{nR}(y) = y+ \dfrac{L_n(\frac{1+y}{2})-L_n(\frac{1-y}{2})}{\sqrt[4]{n}} \text{ for } 0\leq x\leq 1.\] 
	Denote $L^*_n(y):= L_n\left(\dfrac{1+y}{2}\right)-L_n\left(\dfrac{1-y}{2}\right)$ and 
	$L^*(y):= L\left(\dfrac{1+y}{2}\right)-L\left(\dfrac{1-y}{2}\right)$. Now, using Theorem 3 of \cite{sethu1970}, we obtain asymptotic distribution of $W^*(h)$ under the null as well as the local alternatives \eqref{cos:alternative}, as detailed in the following theorems.  
	\begin{theorem}
		The asymptotic distribution of $W^*(h)$ under null hypothesis is given by 
		$$ \frac{1}{\sqrt{n}}\sum_{i=1}^n[h(nDS_i)-\mathbb{E}h(Z)] \stackrel{d}{\rightarrow} N(0,\sigma^2_h) \text{ as }n\to\infty,$$ where $\sigma^2_h= Var(h(Z))- Cov^2(h(Z),Z))$ and $Z$ is a standard exponential random variable.
	\end{theorem}
	\begin{theorem}
		The asymptotic distribution of $W^*(h)$ under the sequence of local alternatives \eqref{cos:alternative} is given by 
		$$ \frac{1}{\sqrt{n}}\sum_{i=1}^n[h(nDS_i)-\mathbb{E}h(Z)] \stackrel{d}{\rightarrow} N(\mu_h,\sigma^2_h) \text{ as }n\to\infty,$$ where $\mu_h= \frac{1}{2}\left(\int_0^1[L^{*\prime}(u)]^2du\right)Cov[h(Z), (Z-2)^2] $ and $Z$ is a standard exponential random variable.
	\end{theorem}
	\subsection{Asymptotic Relative Efficiency under a sequence of Local Alternatives}
	Suppose there are two competing tests corresponding to test statistics $V_n(g_i):=\dfrac{1}{n}\sum_{k=1}^ng_i(DS_k)$ for $i=1,2$. Let $V_n(g_i)$s have asymptotic means zero and finite variances under null hypothesis. Under the sequence of alternatives stated in \eqref{cos:alternative}, let $V_n(g_i)$ have asymptotic mean and variance $\mu(g_i)$ and $\sigma^2(g_i)$, respectively, for $i=1,2$. Then, the Pitman asymptotic relative efficiency (ARE) of $V_n(g_1)$ relative to $V_n(g_2)$ is given by 
	$$ARE(g_1,g_2)= \frac{e^2(g_1)}{e^2(g_2)}= \frac{\left(\frac{\mu^2(g_1)}{\sigma^2(g_1)}\right)^2}{\left(\frac{\mu^2(g_2)}{\sigma^2(g_2)}\right)^2}\,.$$
	The quantity $e(g_i)={\mu^2(g_i)}/{\sigma^2(g_i)}$ is called the efficacy of the test based on $V_n(g_i)$ for $i=1,2$. Under a sequence of local alternatives converging to the null hypothesis, the test with maximum efficacy is asymptotically locally most powerful in terms of the Pitman ARE. \cite{sethu1970} obtained efficacy for tests based on usual spacings $W(h):=\frac{1}{n}\sum _{i=1}^n h(nD_i)$ as below
	$$e(h)= \frac{(\int_0^1l^2(u)du)Cov[h(Z),(Z-2)^2]}{2[Var(h(Z))-Cov^2(h(Z),Z]^{1/2}},$$
	where $l(x):=L'(x)$ and $Z$ is a standard exponential random variable. Similarly, we obtain efficacy of tests based on CO spacings, which is given by the following lemma.
	\begin{lemma}
		For the test statistic $W^*_n(h):=\frac{1}{n}\sum_{k=1}^nh(nDS_k)$, the efficacy under the sequence of alternative \eqref{cos:alternative} is given by
		$$e^*(h)= \dfrac{(\int_0^1l^{*2}(u)du)Cov[h(Z),(Z-2)^2]}{2[var(h(Z))-Cov^2(h(Z),Z]^{1/2}},$$
		where $l^*(x)= l\left(\dfrac{1+x}{2}\right)-l\left(\dfrac{1-x}{2}\right)$ and $Z$ is a standard exponential random variable.
	\end{lemma}
	The following result provides the asymptotically locally most powerful (ALMP) test among tests based on statistics of the type $W^*_n(h)$.
	\begin{theorem}
		For the sequence of alternatives \eqref{cos:alternative}, among tests based on statistics of the type $W^*_n(h)=\frac{1}{n}\sum_{k=1}^nh(nDS_k)$, the test corresponding to $h(x)=x^2$ is most efficient in terms of the Pitman ARE.
	\end{theorem}
	The above theorem is a consequence of a result of \cite{sethu1970}. As expected, the ALMP test is the Greenwood test based on CO spacings. Similar to tests based on usual spacings,  tests based on statistics of the type $W^*_n(h)$ can not detect alternatives converging to the null distribution at a rate faster than $n^{-1/4}$.
	\begin{remark}
		We can define higher order disjoint and overlapping spacings based on CO ordering. For these higher order CO spacings, results similar to those for usual higher order spacings in the existing literature hold true (see, e.g., \citealt{del1979asymptotic,rao-kuo1984,misra2001new}). Also, results similar to those in \cite{tung_2012a,tung_2012b} hold true for CO spacings.
	\end{remark}
	
	\section{Simulation Studies}
	We now perform some simulation studies to assess the finite sample performance of the proposed tests, and compare their performance with tests based on usual spacings. 
	Suppose $GS$, $LS$, $ES$ and $RS$ denote the test statistics corresponding to the Greenwood, log spacing, relative entropy and Rao spacing, respectively, based on usual spacings. Let $GS^*$, $LS^*$, $ES^*$ and $RS^*$ denote test statistics corresponding to the Greenwood, log spacing, relative entropy and Rao spacing, respectively, based on CO spacings. 
	For our study, we take the level of significance to be $0.05$. The empirical powers of the tests are calculated from $10000$ iterates. We consider sample sizes $10,20,30,50,80,100,200$ and $300$.
	\subsection{Uniformity Tests}
	Following Stephens (1974), first we consider alternatives of the following three types (for $k>0$),
	\begin{itemize}
		\item[$A_k$ : ]$F(x)=1-(1-x)^k$, $0\leq x \leq 1$;
		\item[$B_k$ : ]$F(x)= \begin{cases}
			2^{k-1}x^k, & \text{ if }0\leq x\leq 0.5,\\
			1-2^{k-1}(1-x)^k, & \text{ if }0.5\leq x\leq 1;
		\end{cases}$
		\item[$C_k$ : ]$F(x)= \begin{cases}
			0.5-2^{k-1}(0.5-x)^k, & \text{ if }0\leq x\leq 0.5,\\
			0.5+2^{k-1}(x-0.5)^k, & \text{ if }0.5\leq x\leq 1.
		\end{cases}$	
	\end{itemize}
	These families of distribution give a wide variety of dfs supported on $[0,1]$. 
	For $k>1$, the family $A_k$ yields skewed distributions with a cluster near zero, whereas $B_k$ gives symmetric distributions with cluster near $0.5$ and $C_k$ gives symmetric distributions with two clusters near zero and one. Also, the family $B_k$ has lighter tail than the $U(0,1)$ df, whereas the family $C_k$ has heavier tail than the $U(0,1)$ df. 
	For the simulation study, we take $k=1.5$. The empirical powers for various tests are listed in Table \ref{cos:sim:table1}.
	
	\begin{table}[htb!]
		\begin{center}
			\caption{Empirical powers for $A_{1.5}$, $B_{1.5}$ and $C_{1.5}$ alternatives for the $U(0,1)$ null}\label{cos:sim:table1}
			\begin{tabular}{c c c c c c c c c c} 
				\hline
				Alternative & $n$ & G & G* & L & L* & E & E* & R & R* \\ 
				\hline        
				$A_{1.5}$	&	10	&	0.080	&	0.043	&	0.068	&	0.047	&	0.078	&	0.044	&	0.073	&	0.047	\\
				&	20	&	0.112	&	0.050	&	0.084	&	0.049	&	0.105	&	0.047	&	0.090	&	0.047	\\
				&	30	&	0.156	&	0.051	&	0.095	&	0.049	&	0.139	&	0.049	&	0.110	&	0.050	\\
				&	50	&	0.227	&	0.052	&	0.124	&	0.052	&	0.205	&	0.052	&	0.155	&	0.051	\\
				&	80	&	0.338	&	0.053	&	0.172	&	0.053	&	0.306	&	0.055	&	0.220	&	0.057	\\
				&	100	&	0.404	&	0.052	&	0.198	&	0.056	&	0.368	&	0.054	&	0.249	&	0.054	\\
				&	200	&	0.662	&	0.056	&	0.310	&	0.056	&	0.600	&	0.055	&	0.399	&	0.053	\\
				&	300	&	0.821	&	0.059	&	0.414	&	0.054	&	0.753	&	0.058	&	0.527	&	0.059	\\
				\hline																			
				$B_{1.5}$	&	10	&	0.025	&	0.077	&	0.038	&	0.063	&	0.026	&	0.074	&	0.034	&	0.071	\\
				&	20	&	0.055	&	0.116	&	0.059	&	0.081	&	0.058	&	0.110	&	0.059	&	0.092	\\
				&	30	&	0.086	&	0.154	&	0.070	&	0.100	&	0.083	&	0.145	&	0.081	&	0.117	\\
				&	50	&	0.142	&	0.230	&	0.095	&	0.132	&	0.138	&	0.215	&	0.122	&	0.166	\\
				&	80	&	0.247	&	0.339	&	0.137	&	0.166	&	0.232	&	0.304	&	0.175	&	0.217	\\
				&	100	&	0.300	&	0.411	&	0.164	&	0.194	&	0.277	&	0.372	&	0.204	&	0.255	\\
				&	200	&	0.581	&	0.663	&	0.280	&	0.317	&	0.530	&	0.601	&	0.362	&	0.396	\\
				&	300	&	0.765	&	0.824	&	0.388	&	0.414	&	0.707	&	0.752	&	0.506	&	0.527	\\
				\hline																			
				$C_{1.5}$	&	10	&	0.166	&	0.076	&	0.099	&	0.063	&	0.152	&	0.075	&	0.133	&	0.074	\\
				&	20	&	0.215	&	0.116	&	0.119	&	0.080	&	0.192	&	0.111	&	0.145	&	0.099	\\
				&	30	&	0.263	&	0.158	&	0.126	&	0.096	&	0.229	&	0.144	&	0.164	&	0.115	\\
				&	50	&	0.356	&	0.229	&	0.161	&	0.129	&	0.312	&	0.210	&	0.208	&	0.162	\\
				&	80	&	0.466	&	0.336	&	0.206	&	0.168	&	0.407	&	0.301	&	0.261	&	0.216	\\
				&	100	&	0.520	&	0.403	&	0.227	&	0.198	&	0.458	&	0.365	&	0.287	&	0.242	\\
				&	200	&	0.760	&	0.663	&	0.347	&	0.311	&	0.685	&	0.602	&	0.437	&	0.404	\\
				&	300	&	0.875	&	0.821	&	0.449	&	0.413	&	0.810	&	0.756	&	0.564	&	0.534	\\
				
				\hline
			\end{tabular}
		\end{center}
	\end{table}
	Table \ref{cos:sim:table1} suggests that, tests based on usual spacings perform better for alternatives $A_{1.5}$ and $C_{1.5}$, but tests based on CO spacings perform better for alternatives $B_{1.5}$. The alternative  $A_{1.5}$ is not symmetric about $1/2$, and so power of tests based on CO spacings are lower (which is explained by Lemma \ref{cos:lemma2}). For alternatives  $B_{1.5}$ and $C_{1.5}$, powers of tests based on usual spacings and those based on CO spacings are comparable, which is also explained by Lemma \ref{cos:lemma2}. For the light-tailed alternative $B_{1.5}$, tests based on CO spacings are superior to those based on usual spacings.
	
	Next, we consider symmetric $Beta(k,k)$ alternatives with $k=0.5,\,1.5,\,2.5$. The empirical powers of the competing tests are reported in Table \ref{cos:sim:table2}. It is evident that, for the heavy-tailed alternative $Beta(0.5,0.5)$, tests based on usual spacings have better powers. For the light-tailed symmetric alternatives $Beta(1.5,1.5)$ and $Beta(2.5,2.5)$, tests based on CO spacings perform better than those based on usual spacings.
	\begin{table}[htb!]
		\begin{center}
			\caption{Empirical powers for $Beta(k,k)$ alternatives for the $U(0,1)$ null}\label{cos:sim:table2}
			\begin{tabular}{c c c c c c c c c c} 
				\hline
				Alternative & $n$ & G & G* & L & L* & E & E* & R & R* \\ 
				\hline        
				$Beta(0.5,0.5)$	&	10	&	0.253	&	0.205	&	0.418	&	0.343	&	0.312	&	0.253	&	0.301	&	0.247	\\
				&	20	&	0.300	&	0.255	&	0.521	&	0.436	&	0.398	&	0.333	&	0.384	&	0.320	\\
				&	30	&	0.351	&	0.310	&	0.600	&	0.513	&	0.466	&	0.409	&	0.451	&	0.397	\\
				&	50	&	0.444	&	0.414	&	0.722	&	0.652	&	0.591	&	0.549	&	0.586	&	0.536	\\
				&	80	&	0.562	&	0.525	&	0.832	&	0.783	&	0.726	&	0.687	&	0.716	&	0.684	\\
				&	100	&	0.613	&	0.590	&	0.882	&	0.848	&	0.794	&	0.762	&	0.774	&	0.746	\\
				&	200	&	0.840	&	0.833	&	0.980	&	0.975	&	0.949	&	0.947	&	0.941	&	0.938	\\
				&	300	&	0.944	&	0.934	&	0.996	&	0.995	&	0.990	&	0.987	&	0.989	&	0.984	\\
				\hline																			
				$Beta(1.5,1.5)$	&	10	&	0.024	&	0.048	&	0.033	&	0.050	&	0.026	&	0.046	&	0.031	&	0.050	\\
				&	20	&	0.037	&	0.080	&	0.046	&	0.063	&	0.040	&	0.076	&	0.045	&	0.072	\\
				&	30	&	0.057	&	0.101	&	0.051	&	0.072	&	0.055	&	0.093	&	0.056	&	0.080	\\
				&	50	&	0.090	&	0.137	&	0.070	&	0.084	&	0.086	&	0.128	&	0.081	&	0.100	\\
				&	80	&	0.148	&	0.205	&	0.091	&	0.113	&	0.137	&	0.179	&	0.114	&	0.136	\\
				&	100	&	0.174	&	0.252	&	0.111	&	0.130	&	0.166	&	0.220	&	0.128	&	0.151	\\
				&	200	&	0.362	&	0.441	&	0.167	&	0.192	&	0.309	&	0.377	&	0.207	&	0.239	\\
				&	300	&	0.515	&	0.597	&	0.217	&	0.240	&	0.438	&	0.501	&	0.280	&	0.306	\\
				\hline																			
				$Beta(2.5,2.5)$	&	10	&	0.034	&	0.168	&	0.057	&	0.104	&	0.043	&	0.155	&	0.056	&	0.132	\\
				&	20	&	0.149	&	0.377	&	0.117	&	0.184	&	0.160	&	0.344	&	0.152	&	0.249	\\
				&	30	&	0.310	&	0.558	&	0.180	&	0.254	&	0.301	&	0.501	&	0.249	&	0.343	\\
				&	50	&	0.614	&	0.804	&	0.325	&	0.409	&	0.590	&	0.753	&	0.449	&	0.536	\\
				&	80	&	0.881	&	0.950	&	0.511	&	0.592	&	0.848	&	0.925	&	0.664	&	0.725	\\
				&	100	&	0.950	&	0.981	&	0.616	&	0.687	&	0.930	&	0.968	&	0.756	&	0.808	\\
				&	200	&	1.000	&	1.000	&	0.913	&	0.931	&	0.999	&	1.000	&	0.971	&	0.977	\\
				&	300	&	1.000	&	1.000	&	0.983	&	0.987	&	1.000	&	1.000	&	0.997	&	0.997	\\
				\hline
			\end{tabular}
		\end{center}
	\end{table}
	\begin{remark}
		Similar to \cite{li2018}, we can combine tests based on usual spacings and CO spacings, i.e., a test based on $\max(W(h),W^*(h))$. Such tests are also distribution-free. Based on observations of \cite{li2018}, we expect that such tests can perform well for a wide variety of alternatives.
	\end{remark}
	\section{Conclusion}
	In this paper, we have studied several GoF tests based on centre-outward (CO) spacings. New tests are constructed similar to some popular GoF tests based on usual spacings. For a skewed alternative, tests based on CO ordering data have less power compared to those based on the original data. This was also observed by \cite{li2018} in a simulation study. This is explained by the fact that the Hellinger distance decreases for CO ordering based data in the case of skewed alternatives. When the alternative distribution is symmetric and light-tailed, the proposed GoF tests perform better than those based on usual spacings.

	Theoretical results on GoF tests based on higher order spacings extend easily for GoF tests based on higher order CO spacings. There exist studies concerning estimation and parametric tests based on spacings (see, e.g., \citealt{ghosh2001,ekstrom2013}). Such studies based on CO spacings are some potential future problems in this direction.

\bibliographystyle{apalike}
	\bibliography{mybibfile}

\begin{thebibliography}{}

\bibitem[Del~Pino, 1979]{del1979asymptotic}
Del~Pino, G.~E. (1979).
\newblock On the asymptotic distribution of k-spacings with applications to
  goodness-of-fit tests.
\newblock {\em The Annals of Statistics}, pages 1058--1065.

\bibitem[Ekstr\"{o}m, 2013]{ekstrom2013}
Ekstr\"{o}m, M. (2013).
\newblock Powerful parametric tests based on sum-functions of spacings.
\newblock {\em Scand. J. Stat.}, 40(4):886--898.

\bibitem[Ghosh and Jammalamadaka, 2001]{ghosh2001}
Ghosh, K. and Jammalamadaka, S.~R. (2001).
\newblock A general estimation method using spacings.
\newblock {\em J. Statist. Plann. Inference}, 93(1-2):71--82.

\bibitem[Greenwood, 1946]{green}
Greenwood, M. (1946).
\newblock The statistical study of infectious diseases.
\newblock {\em J. Roy. Statist. Soc. (N.S.)}, 109:85--103; discussion,
  103--110.

\bibitem[Li, 2018]{li2018}
Li, J. (2018).
\newblock E{DF} goodness-of-fit tests based on centre-outward ordering.
\newblock {\em J. Nonparametr. Stat.}, 30(4):973--989.

\bibitem[Liu, 1990]{liu1990notion}
Liu, R.~Y. (1990).
\newblock On a notion of data depth based on random simplices.
\newblock {\em The Annals of Statistics}, 18(1):405--414.

\bibitem[Misra and van~der Meulen, 2001]{misra2001new}
Misra, N. and van~der Meulen, E.~C. (2001).
\newblock A new test of uniformity based on overlapping sample spacings.
\newblock {\em Communications in Statistics-Theory and Methods},
  30(7):1435--1470.

\bibitem[Moran, 1951]{moran1951}
Moran, P. A.~P. (1951).
\newblock The random division of an interval. {II}.
\newblock {\em J. Roy. Statist. Soc. Ser. B}, 13:147--150.

\bibitem[Rao and Sethuraman, 1975]{rao1975weak}
Rao, J. and Sethuraman, J. (1975).
\newblock Weak convergence of empirical distribution functions of random
  variables subject to perturbations and scale factors.
\newblock {\em The Annals of Statistics}, 3(2):299--313.

\bibitem[Rao, 1976]{rao:spacing:test}
Rao, J.~S. (1976).
\newblock Some tests based on arc-lengths for the circle.
\newblock {\em Sankhy\={a} Ser. B}, 38(4):329--338.

\bibitem[Rao and Kuo, 1984]{rao-kuo1984}
Rao, J.~S. and Kuo, M. (1984).
\newblock Asymptotic results on the {G}reenwood statistic and some of its
  generalizations.
\newblock {\em J. Roy. Statist. Soc. Ser. B}, 46(2):228--237.

\bibitem[Sethuraman and Rao, 1970]{sethu1970}
Sethuraman, J. and Rao, J.~S. (1970).
\newblock Pitman efficiencies of tests based on spacings.
\newblock In {\em Nonparametric {T}echniques in {S}tatistical {I}nference
  ({P}roc. {S}ympos., {I}ndiana {U}niv., {B}loomington, {I}nd., 1969)}, pages
  405--415. Cambridge Univ. Press, London.

\bibitem[Tukey, 1975]{tukey1975mathematics}
Tukey, J.~W. (1975).
\newblock Mathematics and the picturing of data.
\newblock In {\em Proceedings of the International Congress of Mathematicians,
  Vancouver, 1975}, volume~2, pages 523--531.

\bibitem[Tung and Jammalamadaka, 2012a]{tung_2012b}
Tung, D.~D. and Jammalamadaka, S.~R. (2012a).
\newblock {$U$}-statistics based on higher-order spacings.
\newblock In {\em Nonparametric statistical methods and related topics}, pages
  151--169. World Sci. Publ., Hackensack, NJ.

\bibitem[Tung and Jammalamadaka, 2012b]{tung_2012a}
Tung, D.~D. and Jammalamadaka, S.~R. (2012b).
\newblock {$U$}-statistics based on spacings.
\newblock {\em J. Statist. Plann. Inference}, 142(3):673--684.

\bibitem[Zuo and Serfling, 2000]{zuo2000general}
Zuo, Y. and Serfling, R. (2000).
\newblock General notions of statistical depth function.
\newblock {\em Annals of statistics}, pages 461--482.

\end{thebibliography}
\end{document}